\documentclass[conference]{IEEEtran}
\usepackage[utf8]{inputenc}
\title{Feedback Capacity of Gaussian Channels with Memory}
\author{Oron Sabag, Victoria Kostina, Babak Hassibi}
\date{November 2020}
\usepackage{amsmath, amssymb,amsthm,cite}
\usepackage{graphicx}
\usepackage{psfrag}
\usepackage{bbm}
\usepackage{mathtools}
\usepackage{color}
\usepackage{tcolorbox}
\usepackage{algorithm,algpseudocode}
\usepackage{arydshln}

\DeclareMathOperator*{\nn}{\nonumber}
\DeclareMathOperator{\E}{\mathbb{E}}
\DeclareMathOperator{\vy}{\mathbf{y}}

\DeclareMathOperator{\ve}{\mathbf{e}}
\DeclareMathOperator{\vm}{\mathbf{m}}
\DeclareMathOperator{\vv}{\mathbf{v}}
\DeclareMathOperator{\vw}{\mathbf{w}}
\DeclareMathOperator{\cov}{\mathbf{cov}}
\DeclareMathOperator{\vz}{\mathbf{z}}
\DeclareMathOperator{\vx}{\mathbf{x}}
\DeclareMathOperator{\vs}{\mathbf{s}}
\DeclareMathOperator{\vhs}{\hat{\mathbf{s}}}
\DeclareMathOperator{\vhhs}{\hat{\hat{\mathbf{s}}}}

\allowdisplaybreaks
\newtheorem{lemma}{Lemma}

\newtheorem{theorem}{Theorem}

\theoremstyle{definition}

\newtheorem{assumption}{Assumption}

\newtheorem{conjecture}{Conjecture}

\makeatletter
\def\blfootnote{\gdef\@thefnmark{}\@footnotetext}
\makeatother
\begin{document}

\maketitle
\begin{abstract}
We consider the feedback capacity of a MIMO channel whose channel output is given by a linear state-space model driven by the channel inputs and a Gaussian process. The generality of our state-space model subsumes all previous studied models such as additive channels with colored Gaussian noise, and channels with an arbitrary dependence on previous channel inputs or outputs. The main result is a computable feedback capacity expression that is given as a convex optimization problem subject to a detectability condition. We demonstrate the capacity result on the auto-regressive Gaussian noise channel, where we show that even a single time-instance delay in the feedback reduces the feedback capacity significantly in the stationary regime. On the other hand, for large regression parameters, the feedback capacity can be achieved with delayed feedback. Finally, we show that the detectability condition is satisfied for scalar models and conjecture that it is true for MIMO models.
\end{abstract}
\blfootnote{The authors are with the Department of Electrical Engineering at California Institute of Technology (e-mails:
 \{oron,vkostina,hassibi\}@caltech.edu).}

\section{Introduction}
The feedback capacity of additive Gaussian channels with memory has attracted much interest in information and control theory \cite{Butman69,TiernanSchalk_AR_UB,Ebert,elia_bodemeets,Han_GaussianFeedback,KalmanConnectionISIT20,charalambous2020new}. Notably, \cite{YangKavcicTatikondaGaussian} showed that the capacity problem can be formulated as a Markov decision process, revealing underlying structural relations with control and filtering theory. In \cite{Kim10_Feedback_capacity_stationary_Gaussian}, a finite-dimensional optimization problem was proposed for the case of a colored Gaussian noise generated by a particular state-space model, and in \cite{Gattami} using a change of variable, the optimization problem was shown to be convex with linear matrix inequalities (LMI) constraints. This is an important result but, technically, the change of variable relied on an erroneous claim. In \cite{sabag_MIMO_isit}, the feedback capacity of multiple-input multiple-output (MIMO) channels with colored Gaussian noise has been established. The techniques used in \cite{sabag_MIMO_isit} differ from the ones used in \cite{Kim10_Feedback_capacity_stationary_Gaussian,Gattami}, and they cover a larger class of colored Gaussian noises. Concurrently with \cite{sabag_MIMO_isit}, \cite{Elia_MIMO_ITW} showed that the methodology in \cite{Kim10_Feedback_capacity_stationary_Gaussian,Gattami} for stationary noise processes can be extended to MIMO channels with ISI. Neither \cite{Elia_MIMO_ITW} or \cite{sabag_MIMO_isit} are special cases of one another, but these recent developments on the long-standing capacity problem call for a unified general theory. In this paper, we consider a very general channel formulation that subsumes all studied channels in~\cite{Kim10_Feedback_capacity_stationary_Gaussian,Gattami,sabag_MIMO_isit,Elia_MIMO_ITW} and show that techniques from \cite{sabag_MIMO_isit} can be utilized to conclude the feedback capacity of the general setting.

The channel output of the channel studied in this paper is given by a general state-space model. In particular, the channel is MIMO, and the channel output is a vector $\vy_i\in\mathbb{R}^p$ that is described by the state-space model
\begin{align}\label{eq:lqg_ss}
\vs_{i+1}&= F\vs_i + G \vx_i + \vw_i \nn\\
\vy_i&= H\vs_i + J\vx_i + \vv_i,
\end{align}
where $\vx_i\in\mathbb{R}^m$ is the channel input, $(\vw_i,\vv_i)\sim N\left(0,\begin{pmatrix}W&L\\L^T&V\end{pmatrix}\right)$ is a white Gaussian process and $\vs_i\in\mathbb{R}^n$. The initial state is distributed as $\vs_1\sim N(0,\Sigma_1)$ with $\Sigma_1\succeq0$.

The channel in \eqref{eq:lqg_ss} has memory since past events of the transmission such as previous channel inputs, channel outputs or noise occurrences can affect the current channel output via the hidden state vector $\vs_i$. Several examples are in order.
\begin{enumerate}
    \item The additive white Gaussian noise MIMO channel $$\vy_i = J\vx_i + \vv_i$$ is recovered from \eqref{eq:lqg_ss} with $H=0$.
        \item A channel $\vy_i = J\vx_i + \vz_i$ with an additive colored Gaussian noise process
    \begin{align}
        \vs_{i+1}&= F\vs_i + \vw_i \nn\\
        \vz_i&= H\vs_i  + \vv_i
    \end{align}
    is revealed with $G=0$ in \eqref{eq:lqg_ss}.
    \item A channel with first-order intersymbol interference (ISI)
    $$\vy_i = J\vx_i + H\vx_{i-1} + \vv_i$$
    is revealed with $F=W=L=0$ and $G=I$.
    \item A channel with first-order channel outputs dependence
    $$\vy_i = J\vx_i + H\vy_{i-1} + \vv_i$$
    is revealed from \eqref{eq:lqg_ss} with $F=H$, $G=J$ and $W=V=L$.
    \item Delayed feedback: the state-space in \eqref{eq:lqg_ss} allows one to consider the case where the encoder has access to a delayed feedback. That is, at time $i$, the encoder has only access to the channel outputs $\vy_1,\vy_2,\dots,\vy_{i-d}$ with $d>1$. In \cite{yang_delayed}, it was shown that delayed feedback can be realized with a state-space whose hidden state depends on the channel inputs. This fact has been also utilized in \cite{Elia_MIMO_ITW} and will be demonstrated in Section \ref{sec:main}.
    \item Note that all the variables in \eqref{eq:lqg_ss} are vectors, and therefore, any combination of the special cases above can be realized by \eqref{eq:lqg_ss} as well.
\end{enumerate}

Our main contribution is a solution to the feedback capacity of \eqref{eq:lqg_ss} by showing that it can be formulated as a convex optimization. Our converse proof extends the techniques developed in \cite{sabag_MIMO_isit} for the colored Gaussian noise in $(2)$ to the general channel in $(1)$. The converse technique places only mild regularity conditions on the matrices $F,H,G,J$ in \eqref{eq:lqg_ss} that include, for instance, unstable $F$. The achievability relies on an slight extension of the achievability proof by Cover and Pombra in \cite{Cover89}. A lower bound is then realized by convergence of Riccati recursions.


\section{The setting and Preliminaries}
In this section, we present the setting and preliminaries on Kalman filtering and Riccati equations.

\subsection{The setting and notation}
The channel is given by \eqref{eq:lqg_ss} and a mild assumption on the channel is given in Assumption \ref{ass:1}. At time $i$, the encoder has access to noiseless feedback of the previous channel outputs $\vy^{i-1}\triangleq \{\vy_1,\dots,\vy_{i-1}\}$. Thus, the encoder mapping at time $i$ maps $\vy^{i-1}$ and a message that is distributed uniformly over $[1:2^{nR}]$ to the channel inputs vector in $\mathbb{R}^m$. For a fixed blocklength $n$, the channel input has an average power constraint $\frac1{n} \sum_{i=1}^n \E[\vx_i^T\vx_i]\le P$, and the decoder maps the tuple $\vy^n$ to the message set. Definitions of the average probability of error, an achievable rate and the feedback capacity are standard and can be found, for instance, in \cite{Kim10_Feedback_capacity_stationary_Gaussian}. Throughout the paper, we use $A^\dagger$ to denote the Moore–Penrose inverse of $A$.

\subsection{Kalman filter}
We proceed to present the estimate of the hidden state at the encoder and its recursive computation using a Kalman filter.
\begin{lemma}[Encoder's estimate]\label{lemma:encoder}
The encoder's estimate, defined by $\vhs_i \triangleq \E[\vs_i|\vx^{i-1},\vy^{i-1}]$, can be computed recursively as
\begin{align}\label{eq:enc_recursive}
   \vhs_{i+1}&= F\vhs_{i} + G\vx_i + K_{p,i} (\vy_i - J\vx_i - H\vhs_i)
\end{align}
with the initial condition $\vhs_1=0$. The corresponding error covariance $\Sigma_i\triangleq \cov(\vs_i-\vhs_i)$ can be updated as
\begin{align}\label{eq:enc_recursive_cov}
    \Sigma_{i+1}&= F\Sigma_i F^T + W  - K_{p,i} \Psi_iK_{p,i}^T
\end{align}
with the initial condition $\Sigma_1$. The constants in \eqref{eq:enc_recursive}-\eqref{eq:enc_recursive_cov} are
\begin{align}
    K_{p,i}&= (F\Sigma_iH^T + L)\Psi_i^{-1}\nn\\
    \Psi_i&= H \Sigma_i H^T+V.
\end{align}
Moreover, the innovations process $\ve_i= \vy_i - J\vx_i - H\vhs_i$ is white, distributed according to $\ve_i\sim N(0,\Psi_i)$, and $\ve_i$ is independent of $(\vx^{i},\vy^{i-1})$.
\end{lemma}
The result is standard in Kalman filtering theory, e.g. \cite{kailath_booklinear}. Its proof is omitted due to space limitations.

The asymptotic behaviour of the sequence $\Psi_i$ determines the feedback capacity value. Here, we provide a sufficient condition to guarantee that $\Sigma_i$ converges to the maximal solution, denoted by  $\Sigma$, of the Riccati equation
\begin{align}\label{eq:riccati_general}
    \Sigma&= F \Sigma F^T  + W- K_p\Psi K_p^T,
\end{align}
where $K_p = (F\Sigma H^T + L) \Psi^{-1}$ and $\Psi = H\Sigma H^T + V$.
\begin{assumption}\label{ass:1}
The pair $(F,H)$ is detectable. That is, there exists a matrix $K$ such that $\rho(F-KH)<1$ where $\rho(\cdot)$ is the spectral radius of a matrix. We also assume that $\Sigma_1\succeq \Sigma$.
\end{assumption}
Assumption \ref{ass:1} is sufficient for $\Sigma_i$ to converge to the maximal solution of the Riccati equation~\cite{convergence_detectable_initial}. Note that if the matrix $F$ is stable, the detectability condition is satisfied with $K=0$.

There are several methods to check the detectability of a pair of matrices. As detectability is also part of the capacity solution in Theorem \ref{th:main} below, we present here a simple numerical method using LMI optimization \cite[Sec. $3.12.2$]{caverly2019lmi}: a pair $(A,B)$ is detectable iff there exists $P\succ0$ such that
\begin{align}\label{eq:LMI_detec}
    \begin{pmatrix}
        P & PA\\
        A^TP& P + B^TB
    \end{pmatrix}\succ0.
\end{align}

\section{Feedback Capacity}\label{sec:main}
The following is our main result.
\begin{theorem}[MIMO channels]\label{th:main}
The feedback capacity of the Gaussian channel with memory in \eqref{eq:lqg_ss} with a power constraint $P$, subject to Assumption \ref{ass:1}, is equal to the convex optimization
\begin{align}\label{eq:th_main}
    &\max_{\Gamma,\Pi,\hat{\Sigma}} \ \ \frac1{2}\log \det (\Psi_{Y}) - \frac1{2}\log\det(\Psi)\nn\\
&\ s.t. \ \begin{pmatrix}
    \Pi & \Gamma\\
    \Gamma^T& \hat{\Sigma}
\end{pmatrix} \succeq0, \begin{pmatrix}
     \Omega & K_{Y}\Psi_{Y} \\
    \Psi_{Y} K_{Y}^T & \Psi_{Y}
\end{pmatrix}\succeq0, \ \mathbf{Tr}(\Pi)\le P, \nn\\
&    \Psi_{Y}= H \hat{\Sigma} H^T + J  \Pi J^T +H \Gamma^T J^T + J\Gamma H^T + \Psi \\
&\Omega = F \hat{\Sigma} F^T + G \Pi G^T  + G\Gamma F^T  +  F \Gamma^TG^T + K_p \Psi K_p^T - \hat{ \Sigma}\nn\\
& K_{Y}= ( F \hat{ \Sigma}H^T + F \Gamma^TJ^T + G\Gamma H^T + G\Pi J^T + K_{p}\Psi) \Psi_{Y}^{-1},\nn
\end{align}
if the optimal pair $(F +G\Gamma \hat{\Sigma}^\dagger,H +J\Gamma \hat{\Sigma}^\dagger)$ is detectable.
\end{theorem}

\begin{figure}[b]
    \centering
\psfrag{X}[t][][1]{The regression parameter $\beta$}
\psfrag{Y}[b][][1]{Rate [bits/ch. use]}
\psfrag{M}[l][][0.52]{Feedback capacity}
\psfrag{N}[l][][0.55]{$2$-delay feedback}
\psfrag{O}[l][][0.55]{$3$-delay feedback }
\psfrag{P}[l][][0.55]{$4$-delay feedback }
\psfrag{Q}[l][][0.55]{No feedback}
\psfrag{T}[][][0.8]{Capacity of the Auto-Regressive Gaussian Noise Channel}

\includegraphics[scale=0.19]{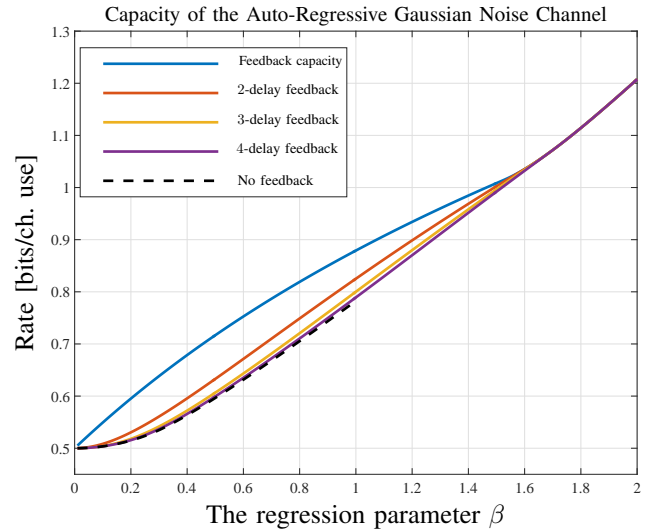}
    \caption{The feedback capacity of the Gaussian channel with a first-order auto-regressive noise and a unit-power constraint. The dashed-black curve describes the no-feedback capacity computed via a water-filling solution and the other curves correspond to the feedback capacity with various delays.}
    \label{fig:AR_capa}
\end{figure}
The objective is a concave function as $\Psi_Y$ is a linear function of $(\Gamma,\Pi,\hat{\Sigma})$. Moreover, the matrix inequalities include only linear terms of the decision variables. Thus, this convex optimization can be carried out with CVX \cite{cvx}. All previous capacity solutions, e.g. \cite{sabag_MIMO_isit,Elia_MIMO_ITW}, are special cases of Theorem~\ref{th:main}.

\subsection{Example: Auto-regressive noise}
The auto-regressive (AR) process of first order is given by \begin{align}\label{eq:AR_discussion}
    z_i&= \beta z_{i-1} + w_i,
\end{align}
where $w_i\sim N(0,1)$ is an i.i.d. sequence and $\beta$ is the regression parameter. This is one of the simplest instances of colored Gaussian noise \cite{Butman69,TiernanSchalk_AR_UB}. Its feedback capacity was derived in \cite{Kim10_Feedback_capacity_stationary_Gaussian} for $|\beta| <1$, and in \cite{sabag_MIMO_isit} for arbitrary $\beta$. In our general setting, we can study arbitrary $\beta$, and we also evaluate the delayed-feedback capacity \cite{yang_delayed,Elia_MIMO_ITW}. For instance, if we consider the scenario where the channel input $\vx_i$ only depends on the delayed-feedback tuple $\vy^{i-2}$, we can realize this scenario with the state-space model
\begin{align}\label{eq:delay_ss}
\vs_{i+1}&= \begin{pmatrix}
        \beta&0\\0&0
    \end{pmatrix} \vs_i + \begin{pmatrix}
    0\\1
    \end{pmatrix} x_i + \vw_i \nn\\
y_i&= \begin{pmatrix}
        \beta & 1
    \end{pmatrix}\vs_i + 0\cdot x_i + v_i,
\end{align}
with $(\vw_i,v_i)\sim N(0,\Lambda )$ with $\Lambda = \left(\begin{array}{cc;{2pt/2pt}c}
    1&0&1\\
    0&0&0\\ \hdashline[2pt/2pt]
    1&0&1
\end{array}\right)$.
The main idea is that the current channel input $x_i$ does not affect the current channel output. Instead, $\vx_i$ is stored as the second element of the hidden state $\vs_{i+1}$ in order to affect the channel output at the next time. This simple idea can be extended to arbitrary, but finite delay, and the corresponding delayed-feedback capacity can be directly computed with the convex optimization in Theorem~\ref{th:main}.

In Fig. \ref{fig:AR_capa}, the feedback capacity for various delays, as well as the no-feedback capacity (using water-filling) are plotted for $P=1$. It can be seen that the delayed-feedback curves approach the feedback capacity for growing regression parameter $\beta$. The largest capacity reduction occurs when the feedback capacity with an instantaneous feedback has an additional time-instance delay, that is, the channel input depends on $\vy^{i-2}$ rather than $\vy^{i-1}$.

\subsection{On the detectability condition}
The detectability condition on the optimal matrix pair in Theorem \ref{th:main} is needed for the lower bound. That is, we show that the convex optimization in \eqref{eq:th_main} is always an upper bound on the capacity, and is achievable when the condition is met. Recall that the detectability condition can be verified using the LMI optimization \eqref{eq:LMI_detec}. Based on extensive simulations, we conjecture the following.
\begin{conjecture}\label{conj:det}
The detectability condition in Theorem \ref{th:main} is redundant. That is, if $(\Gamma,\Pi,\hat{\Sigma})$ solves the optimization in \eqref{eq:th_main}, then $(F +G\Gamma \hat{\Sigma}^\dagger,H +J\Gamma \hat{\Sigma}^\dagger)$ is detectable.
\end{conjecture}
\begin{proof}[Proof of Conjecture \ref{conj:det} for the scalar model] Assume by contradiction that the optimal pair is not detectable, that is, $F +G\Gamma \hat{\Sigma}^\dagger > 1$ and $H+J\Gamma\hat{\Sigma}^\dagger=0$. By the detectability of $(F,H)$ in Assumption \ref{ass:1}, the optimal variables satisfy $\Gamma,\hat{\Sigma}\neq0$.
The proof now follows by writing $\Psi_Y$ as
\begin{align}
    (H+J\Gamma\hat{\Sigma}^\dagger)\hat{\Sigma}(H+J\Gamma\hat{\Sigma}^\dagger)^T + J(P - \Gamma\hat{\Sigma}^\dagger\Gamma^T)J^T + \Psi,\nn
\end{align}
where we used the fact that in the scalar case $\Pi = P$; if $H+J\Gamma\hat{\Sigma}^\dagger=0$, choosing $\Gamma=0$ (and a feasible $\hat{\Sigma}$) attains a larger objective for $J\neq0$ contradicting the assumed optimality. Note that if $J=0$, the assumption $H+J\Gamma\hat{\Sigma}^\dagger=0$ implies $H=0$ and the channel capacity is zero.
\end{proof}
\section{Proof of the main result}
In this section, we prove Theorem \ref{th:main}. The capacity derivation utilizes the finite-horizon optimization of directed information
\begin{align}\label{eq:DI_n}
C_n(P) \triangleq \max_{ P(\vx^n||\vy^n): \frac1{n}\sum_i \E[\vx^T_i\vx_i] \le P} I(\vx^n\to \vy^n).
\end{align}
Its relation to the feedback capacity is shown explicitly when needed. While we study directly the optimization problem in \eqref{eq:DI_n}, it is convenient to use the terminology of an encoder and a decoder to describe the solution to the problem. We will use two estimates
\begin{align}
    \vhs_i &= \E[\vs_i|\vx^{i-1},\vy^{i-1}]\nn\\
    \vhhs_i&= \E[\vhs_i|\vy^{i-1}],
\end{align}
where the first is referred to as the \emph{encoder's estimate} since at time $i$ the encoder has access to previous channel inputs and outputs, while the second estimate, $\vhhs_i$, is referred to as the decoders' estimate since it is based on channel outputs only.

\subsection{Optimal policy}
By Lemma \ref{lemma:encoder}, the channel outputs can be described as the linear dynamical system
\begin{align}\label{eq:proof_ss_1}
    \vhs_{i+1}&= F\vhs_i+ G\vx_i + K_{p,i} \ve_i \nn \\
    \vy_i&= H \vhs_i + J\vx_i +  \ve_i,
\end{align}
where the innovation process definition in Lemma \ref{lemma:encoder} was used as $H\vs_i + \vv_i = \ve_i + H\vhs_i$. Special cases of the state-space structure above appeared in   \cite{YangKavcicTatikondaGaussian,Kim10_Feedback_capacity_stationary_Gaussian,charalambous2020new,sabag_MIMO_isit,Elia_MIMO_ITW}. The next step is to show that the maximization domain can be simplified.

\begin{lemma}[The optimal policy structure]\label{lemma:policy}
For a fixed $n$, it is sufficient to optimize the directed information in \eqref{eq:DI_n} over inputs of the form
\begin{align}\label{eq:policy}
\vx_i&= \Gamma_i \hat{\Sigma}^\dagger_{i}(\vhs_{i}-\vhhs_{i}) + \vm_i,\ \ \ \  i=1,\dots,n
\end{align}
where $\vm_i\sim N(0,M_i)$ is independent of $(\vx^{i-1},\vy^{i-1})$, $\hat{\Sigma}^\dagger_{i}$ is the Moore-Penrose pseudo-inverse of
\begin{align}
    \hat{\Sigma}_{i}&= \cov( \vhs_{i}-\vhhs_{i}),
\end{align}
$\Gamma_i$ is a matrix that satisfies
\begin{align}\label{eq:lemma_policy_orthognality}
\Gamma_i(I - \hat{\Sigma}_{i}\hat{\Sigma}^\dagger_{i}) =0,
\end{align}
and the power constraint is
\begin{align}\label{eq:lemma_policy_pow}
    \frac1{n}\sum_{i=1}^n\mathbf{Tr}( \Gamma_i \hat{\Sigma}^\dagger_{i}\Gamma_i^T + M_i)\le P.
\end{align}
\end{lemma}
\begin{proof}[Proof of Lemma \ref{lemma:policy}]
We show here that
\begin{align}\label{eq:cond_const}
    h(\vy_i|\vx^i,\vy^{i-1})&= \frac1{2}\log \det \Psi_i + \frac{p}{2}\log (2\pi e),
\end{align}
for any inputs distribution $P(\vx^n\|\vy^n)$, where $\Psi_i$ is the constant given in Lemma \ref{lemma:encoder}. This implies that the optimization of the directed information in \eqref{eq:DI_n} boils down to maximizing the entropy rate of channel outputs. The proof then follows from \cite[Lemma $1$]{sabag_MIMO_isit}, whose details are omitted here.

Consider the following chain of equalities
\begin{align}
&h(\vy_i|\vx^i,\vy^{i-1})= h(H \vs_i + \vv_i |\vx^i,\vy^{i-1})\nn\\
&\stackrel{(a)}= h(H \vs_i + \vv_i |\vx^{i-1},\vy^{i-1})\nn\\
&\stackrel{(b)}= h_{\mathbb G}(H \vs_i + \vv_i |\vx^{i-1},\vy^{i-1})\nn\\
&\stackrel{(c)} = \frac1{2} \log\det( H\text{cov}(\vs_i - \vhs_i) H^T + V) + \frac{p}{2}\log (2\pi e)\nn\\
&= \frac1{2} \log\det \Psi_i+ \frac{p}{2}\log (2\pi e),
\end{align}
where $(a)$ follows from the Markov chain $\vs_i - (\vx^{i-1},\vy^{i-1})-\vx_i$ and the fact that $\vv_i$ is independent from past occurrences, $(b)$ follows from the fact that $\vs_i$ conditioned on $(\vx^{i-1},\vy^{i-1})$ is an affine function of the Gaussian vectors $\vw^{t}$ and $\vs_0$, and $(c)$ follows from the Gaussian distribution of $\vs_i|\vx^{i-1},\vy^{i-1}$ and the fact that $\vv_i$ is independent from $(\vx^{i-1},\vy^{i-1})$.
\end{proof}
\subsection{The decoder estimate}
Using the optimal policy in Lemma \ref{lemma:policy}, we can write \eqref{eq:proof_ss_1} as
\begin{align}\label{eq:ss_dec}
    \vhs_{i+1}&= (F +G\Gamma_i \hat{\Sigma}_i^\dagger) \vhs_i - G\Gamma_i \hat{\Sigma}_i^\dagger \vhhs_i + G\vm_i + K_{p,i} \ve_i\nn\\
    \vy_i&= (H +J\Gamma_i \hat{\Sigma}_i^\dagger)\vhs_i - J\Gamma_i \hat{\Sigma}_i^\dagger \vhhs_i + J \vm_i +  \ve_i.
\end{align}
The Kalman filter analysis of \eqref{eq:ss_dec} is given as follows.
\begin{lemma}\label{lemma:ss_dec}
For the state space in \eqref{eq:ss_dec}, the optimal estimator is given by
\begin{align}
    \vhhs_{i+1}
    &= F \vhhs_i  + K_{Y,i}(\vy_i  - H\vhhs_i),
\end{align}
with $\vhhs_1=0$, its error covariance $\hat{\Sigma}_i\triangleq \cov (\vhs_i-\vhhs_i)$ is given by
\begin{align}\label{eq:lemma_riccati_recur}
     \hat{\Sigma}_{i+1}&= (F + G\Gamma_i \hat{\Sigma}_i^\dagger) \hat{\Sigma}_i (F + G\Gamma_i \hat{\Sigma}_i^\dagger)^T \nn\\
    &\ \ + GM_i G^T + K_{p,i} \Psi_i K_{p,i}^T  - K_{Y,i} \Psi_{Y,i} K_{Y,i}^T,
\end{align}
with $\hat{\Sigma}_1=0$ and the constants
\begin{align}\label{eq:lemma_ss_y_psiK}
\Psi_{Y,i}&= \mspace{-3mu}(H \mspace{-3mu}+ \mspace{-3mu} J\Gamma_i \hat{\Sigma}_i^\dagger) \hat{\Sigma}_i (H \mspace{-3mu}+\mspace{-3mu}J\Gamma_i \hat{\Sigma}_i^\dagger)^T + J M_i J^T + \Psi_i\nn\\
K_{Y,i} &= \left[(F +G\Gamma_i \hat{\Sigma}_i^\dagger)\hat{\Sigma}_i(H +J\Gamma_i \hat{\Sigma}_i^\dagger)^T \right.\nn\\
& \ \ \ \ \ \ \left.+ G M_iJ^T + K_{p,i}\Psi_i \right]\Psi_{Y,i}^{-1}.
\end{align}
\end{lemma}
The proof is a standard Kalman filtering analysis with the slight modification of the appearance of the variable $\vhhs_i$ which is a function of $\vy^{i-1}$. This modification has no effect on the error covariance $\hat{\Sigma}_i$, but it affects the recursive update of $\vhhs_i$. By \eqref{eq:ss_dec} and Lemma \ref{lemma:ss_dec}, it can be concluded that the measurements entropy rate is
\begin{align}\label{eq:dec_entr}
    h(\vy^n)&= \frac1{2}\sum_{i=1}^n \log\det \Psi_{Y,i} + \frac{p}{2}\log (2\pi e).
\end{align}

\subsection{Upper bound}
\begin{lemma}\label{lemma:scop}
The maximal directed information in \eqref{eq:DI_n} is upper bounded by the convex optimization problem
\begin{align}\label{eq:lemma_SCOP}
    &C_n(P)\le \max_{\{\Gamma_i,\Pi_i,\hat{ \Sigma}_{i+1}\}_{i=1}^n} \ \ \frac1{2}\sum_{i=1}^n\log \det (\Psi_{Y,i}) - \log\det(\Psi_i)\nn\\
&\ \ \ \ \ s.t. \ \ \ \begin{pmatrix}
    \Pi_i & \Gamma_i\\
    \Gamma_i^T& \hat{ \Sigma}_i
\end{pmatrix} \succeq0, \ \ \frac1{n}\sum_{i=1}^n\mathbf{Tr}(\Pi_i)\le P,  \nn\\
&\Psi_{Y,i}= H \hat{\Sigma}_i H^T + J  \Pi_i J^T +H \Gamma_i^T J^T + J\Gamma_i H^T + \Psi_i \nn \\
&\begin{pmatrix}
     \Omega_i & K_{Y,i}\Psi_{Y,i} \\
    \Psi_{Y,i} K_{Y,i}^T & \Psi_{Y,i}
    \end{pmatrix}\succeq0,\nn\\
& K_{Y,i} \Psi_{Y,i}= F \hat{\Sigma}_iH^T \mspace{-3mu}+ \mspace{-3mu} F \Gamma_i^T J^T \mspace{-3mu}+ \mspace{-3mu} G\Gamma_i H^T \mspace{-3mu}+ \mspace{-3mu} G \Pi_i J^T \mspace{-3mu}+ \mspace{-3mu} K_{p,i}\Psi_i \nn \\
& \Omega_i = F\hat{\Sigma}_iF^T - \hat{\Sigma}_{i+1}+ G\Gamma_i F^T + F\Gamma_i^TG^T  + G \Pi_i G^T \nn\\
&\ \ \ \ \ + K_{p,i} \Psi_i K_{p,i}^T,
\end{align}
where the constraints hold for $i=1,\dots,n$ and $\hat{\Sigma}_1= 0$.
\end{lemma}
\begin{proof}[Proof of Lemma \ref{lemma:scop}]
Using \eqref{eq:dec_entr} and \eqref{eq:cond_const}, the directed information for policies of the form \eqref{eq:policy} can be written as
\begin{align}
    C_n(P)&= \max \frac1{2}\sum_{i=1}^n \log\det(\Psi_{Y,i}) - \log\det \Psi_i,
\end{align}
where the maximization domain is $\{\Gamma_i,M_i\succeq0\}_{i=1}^n$ subject to the constraints in \eqref{eq:lemma_policy_orthognality}-\eqref{eq:lemma_policy_pow}, the Riccati recursion in \eqref{eq:lemma_riccati_recur} and $\Psi_{Y,i},K_{Y,i}$ are given in \eqref{eq:lemma_ss_y_psiK}.

The first step is to define the auxiliary decision variable $$\Pi_i \triangleq \Gamma_i \hat{\Sigma}^\dagger_{i}\Gamma_i^T + M_i.$$ Using the orthogonality constraint in \eqref{eq:lemma_policy_orthognality}, the variables in \eqref{eq:lemma_ss_y_psiK} can be expressed as
\begin{align}
    \Psi_{Y,i} &\mspace{-2mu}= \mspace{-2mu}H \hat{\Sigma}_i H^T + J  \Pi_i J^T +H \Gamma_i^T J^T + J\Gamma_i H^T + \Psi_i\nn\\
    K_{Y,i}\Psi_{Y,i} &\mspace{-2mu}=\mspace{-2mu} F \hat{\Sigma}_iH^T \mspace{-3mu}+ \mspace{-3mu}F \Gamma_i^T J^T \mspace{-3mu}+ \mspace{-3mu} G\Gamma_i H^T \mspace{-3mu}+ \mspace{-3mu} G \Pi_i J^T \mspace{-3mu}+ \mspace{-3mu} K_{p,i}\Psi_i,\nn
\end{align}
and the Riccati recursion in \eqref{eq:lemma_riccati_recur} can be expressed as
\begin{align}\label{eq:proof_scop_recursion}
     0&= \Omega_i - K_{Y,i} \Psi_{Y,i} K_{Y,i}^T,
\end{align}
with
\begin{align}
    \Omega_i&\triangleq F\hat{\Sigma}_iF^T - \hat{\Sigma}_{i+1}+ G\Gamma_i F^T + F\Gamma_i^TG^T  + G \Pi_i G^T \nn\\
    &\ \ \ + K_{p,i} \Psi_i K_{p,i}^T.
\end{align}
We now relax \eqref{eq:proof_scop_recursion} into an inequality and, due to $\Psi_{Y,i}\succ0$, that inequality is equivalent using a Schur complement to
\begin{align}\label{eq:proof_scop_RiccLMI}
    \begin{pmatrix}
    \Omega_i& K_{Y,i}\Psi_{Y,i}\\
    \Psi_{Y,i}K_{Y,i}^T& \Psi_{Y,i}
    \end{pmatrix}\succeq0.
\end{align}
Note that the variable $M_i\succeq0$ appears only in the constraint $\Pi_i =\Gamma_i \hat{\Sigma}^\dagger_{i}\Gamma_i^T + M_i$, so the constraint can be reduced to the inequality $\Pi_i -\Gamma_i \hat{\Sigma}^\dagger_{i}\Gamma_i^T \succeq0$. Finally, we use the Schur complement for the positive semidefinite matrix $\hat{\Sigma}_i\succeq0$ along with the inequality $\Pi_i -\Gamma_i \hat{\Sigma}^\dagger_{i}\Gamma_i^T \succeq0$ and the orthogonality constraint in \eqref{eq:lemma_policy_orthognality} to combine them into a single LMI
\begin{align}
        \begin{pmatrix}
    \Pi_i& \Gamma_i\\
    \Gamma_i^T& \hat{\Sigma}_i
    \end{pmatrix}\succeq0.
\end{align}
\end{proof}
The upper bound in Theorem~\ref{th:main} is shown next.
\begin{theorem}[Feedback capacity upper bound]\label{th:UB}
The feedback capacity is upper bounded by the convex optimization
\begin{align}\label{eq:lemma_UB}
    &C_{fb}(P)\le \max_{\Gamma,\Pi,\hat{\Sigma}} \ \ \frac1{2}\log \det (\Psi_{Y}) - \frac1{2}\log\det(\Psi)\nn\\
&\ \text{s.t.} \ \begin{pmatrix}
     \Omega& K_{Y}\Psi_{Y} \\
    \Psi_{Y} K_{Y}^T & \Psi_{Y}
    \end{pmatrix}\succeq0, \begin{pmatrix}
    \Pi & \Gamma\\
    \Gamma^T& \hat{\Sigma}
\end{pmatrix} \succeq0, \ \mathbf{Tr}(\Pi)\le P, \nn\\
&\Psi_{Y}= H \hat{\Sigma} H^T + J  \Pi J^T +H \Gamma^T J^T + J\Gamma H^T + \Psi  \\
& K_{Y}= (F \hat{\Sigma}H^T + F \Gamma^T J^T + G\Gamma H^T + G \Pi J^T + K_{p}\Psi) \Psi_{Y}^{-1}\nn \\
& \Omega = F \hat{\Sigma} F^T + G \Pi G^T  + G\Gamma F^T  +  F \Gamma^TG^T + K_p \Psi K_p^T - \hat{ \Sigma}.\nn
\end{align}
\end{theorem}
\begin{proof}[Proof of Theorem \ref{th:UB}]
Let $R$ be an achievable rate achieved by a sequence of codebooks, and consider the upper bound
\begin{align}\label{eq:proof_UB_fano}
    nR &= H(M)\nn\\
    &\stackrel{(a)}\le \sum_{i=1}^n I(\vx^n\to \vy^n) + \epsilon_n\nn\\
    &\stackrel{(b)}\le C_n(P) + \epsilon_n,
\end{align}
where $(a)$ follows from Fano's inequality with $\epsilon_n\to0$ and the fact that $\vx_i$ is a function of the message $M$ and the feedback tuple $\vy^{i-1}$, and $(b)$ follows from the definition of $C_n(P)$ in \eqref{eq:DI_n} and by taking a maximum over all power-admissible channel inputs. We can now utilize the upper bound on $C_n(P)$ in Lemma \ref{lemma:scop} to derive the single-letter upper bound in Theorem~\ref{th:UB}. As the proof follows similar steps to \cite[Lemma $4$]{sabag_MIMO_isit}, it is omitted here due to space limitations.

\end{proof}

\subsection{Achievability (Lower bound)}
\begin{lemma}[Lower bound]\label{lemma:achievable}
The following optimization is a lower bound to the feedback capacity
\begin{align}
    &C_{fb}(P)\ge \max_{\Gamma,\Pi,\hat{\Sigma}} \ \ \frac1{2}\log \det (\Psi_{Y}) - \frac1{2}\log\det(\Psi)\nn\\
&\ \text{s.t.} \ \begin{pmatrix}
     \Omega& K_{Y}\Psi_{Y} \\
    \Psi_{Y} K_{Y}^T & \Psi_{Y}
    \end{pmatrix}\succeq0, \begin{pmatrix}
    \Pi & \Gamma\\
    \Gamma^T& \hat{\Sigma}
\end{pmatrix} \succeq0, \ \mathbf{Tr}(\Pi)\le P, \nn\\
&\Psi_{Y}= H \hat{\Sigma} H^T + J  \Pi J^T +H \Gamma^T J^T + J\Gamma H^T + \Psi  \\
& K_{Y}= (F \hat{\Sigma}H^T + F \Gamma^T J^T + G\Gamma H^T + G \Pi J^T + K_{p}\Psi) \Psi_{Y}^{-1}\nn \\
& \Omega = F \hat{\Sigma} F^T + G \Pi G^T  + G\Gamma F^T  +  F \Gamma^TG^T + K_p \Psi K_p^T - \hat{ \Sigma}.\nn\\
 &\exists K: \rho(F +G\Gamma \hat{\Sigma}^\dagger - K(H +J\Gamma \hat{\Sigma}^\dagger))<1. \label{eq:constraint_detec}
\end{align}
\end{lemma}
\begin{proof}[Proof of Lemma \ref{lemma:achievable}]
In \cite{Cover89}, Cover and Pombra proved an achievability bound for an additive channel of the form $\vy^n = \vx^n + \vz^n$ where $\{\vz_i\}$ is a Gaussian process with a given covariance. In \cite[Sec. VII]{Cover89}, it is shown that if the sequence $\{\vx\}$ is a sum of a strictly causal function of the sequence $\{\vz\}$, and an additional independent Gaussian sequence, any rate $R < \frac1{n} [h(\vy^n) - h(\vz^n)]$ is achievable for large enough $n$. Our setting does not fall into this problem formulation since the channel outputs cannot be written as the sum of the channel inputs and an independent noise sequence. More explicitly, if write our channel outputs as $\vy^n = \vx^n + \vz^n$ then the resulting sequence $\vz_i$ will be affected by the channel inputs in a strictly causal manner. The proof essentially shows that Cover and Pombra argument can be generalized to our general scenario so that the averaged directed information is achievable. We then apply our coding methods from \cite{sabag_MIMO_isit} to arrive at the optimization problem in Lemma \ref{lemma:achievable}.

Let $\vz_i \triangleq \vy_i - J\vx_i = H\vs_i + \vv_i$ be the \emph{effective channel noise}, so we can write the channel as $\vy^n = J_n\vx^n + \vz^n$ where $J_n$ is a block-diagonal matrix with the matrix $J$ on its main diagonal. For our particular policy in Lemma \ref{lemma:policy} $\vx_i = \Gamma(\vhs_i - \vhhs_i) + \vm_i $, the sequence $\{\vx_i\}$ is the sum of a strictly causal function of the effective noise sequence $\{\vz_i\}$, and a causal function of the sequence $\{\vm_i\}$. In other words, we can write as a short hand $$\vx_i = B_i\vz^{i-1} + C_i\vm^i,$$ for some matrices $B_i,C_i$ that are strictly causal (strictly lower-triangular) and causal (lower triangular), respectively. The proof steps in \cite[Sec. VII]{Cover89} follow directly with our modification to the effective channel noise expect for a particular chain of inequalities in their error analysis. Namely, \cite[Eq. $64$]{Cover89} begins as
\begin{align}\label{eq:cover_64}
    \Pr \{E_2|W=1\}&= \int_{(m^n,y^n)\in A_\epsilon^n} f(m^n)f(y^n)dm^n dy^n\nn\\
    &\le 2^{-n(h(\vy^n) - h(\vy^n|\vm^n)+3\epsilon)},
\end{align}
where we changed $v_i$ to $\vm_i$ in our notation. The remaining steps are now modified for our setting; the conditional term in \eqref{eq:cover_64} can be simplified as
\begin{align}
    h(\vy^n|\vm^n)&= \sum_{i=1}^n h(\vy_i|\vy^{i-1},\vm^n)\nn\\
    &\stackrel{(a)}= \sum_{i=1}^n h(\vz_i|\vx^i,\vy^{i-1},\vm^n)\nn\\
    &\stackrel{(b)}= \sum_{i=1}^n h(\vz_i|\vz^{i-1},\vx^{i-1})\nn\\
    &=\frac1{2} \sum_{i=1}^n \log\det \Psi_i + \frac{p}{2}\log (2\pi e),
\end{align}
where $(a)$ follows from the fact that $\vx^i$ is a causal function of $\{\vm_i\}$ and $(b)$ follows from the fact that, given $\vx^i$, the sufficient statistic in the state-space \eqref{eq:lqg_ss} is $\vz^{i-1}$. This completes the proof that $R< \frac1{n}[h(\vy^n) - \sum_{i=1}^n \log\det \Psi_i]$ is achievable for large enough $n$. By \eqref{eq:cond_const}, this rate is equal to the directed information in \eqref{eq:DI_n}, and by Lemma \ref{lemma:policy} the policy considered here is optimal leading to the conclusion that $R<\frac1{n}C_n(P)$ is achievable.

Given the state-space representation in Lemma \ref{lemma:ss_dec}, the rest of the proof follows from steps similar to \cite[Lemma $5-6$]{sabag_MIMO_isit} to arrive at the optimization problem in Theorem \ref{th:main} under the assumption that $(F +G\Gamma \hat{\Sigma}^\dagger,H +J\Gamma \hat{\Sigma}^\dagger)$ is detectable.
\end{proof}

\bibliography{ISIT22}
\bibliographystyle{IEEEtran}

%

\end{document}